\documentclass[a4paper,11pt]{article}
\usepackage[utf8]{inputenc}
\usepackage{amsmath}
\usepackage{amsfonts}
\usepackage{amsthm}
\usepackage{subcaption}

\usepackage[english]{babel}
\usepackage{amsmath}
\usepackage{amsfonts}
\usepackage{amssymb}
\usepackage{color,colordvi}
\usepackage{rotating}

\theoremstyle{plain}
\newtheorem{theorem}{Theorem}[section]
\newtheorem{lemma}[theorem]{Lemma}
\theoremstyle{definition}
\newtheorem{definition}[theorem]{Definition}
\usepackage{xspace}
\usepackage{euscript}
\usepackage{graphicx}
\usepackage{amscd}
\usepackage{tabularx}

\usepackage{enumerate}
\usepackage{epsfig} 
\usepackage{graphics} 
\newcommand{\R}{\mathbb{R}}
\newcommand{\N}{\mathbb{N}}

\usepackage{booktabs}
\usepackage{multirow}

\numberwithin{equation}{section}

\title{A uniqueness result on detecting a prey in a spider orb-web}
\author{Alexandre Kawano\thanks{Escola Politecnica,  University of S\~ao Paulo, S\~ao Paulo, Brazil. E-mail:
\textsf{akawano@usp.br}} \and Antonino Morassi\thanks{Polytechnic
Department of Engineering and Architecture, Universit\`a degli
Studi di Udine, Udine, Italy. E-mail:
\textsf{antonino.morassi@uniud.it}}  \ \ -- \ \ 
May 16, 2019  }

\begin{document}

\maketitle

\begin{abstract}
We consider the inverse problem of localizing a prey hitting a
spider orb-web from dynamic measurements taken near the center
of the web, where the spider is supposed to stay. The actual
discrete orb-web, formed by a finite number of radial and
circumferential threads, is modelled as a continuous membrane. The
membrane has a specific fibrous structure, which is inherited from
the original discrete web, and it is subject to tensile pre-stress
in the referential configuration. The transverse load describing
the prey's impact is assumed of the form $g(t)f(x)$, where $g(t)$
is a known function of time and $f(x)$ is the unknown term
depending on the position variable $x$. For axially-symmetric
orb-webs supported at the boundary and undergoing infinitesimal
transverse deformations, we prove a uniqueness result for $f(x)$
in terms of measurements of the transverse dynamic displacement taken
on an arbitrarily small and thin ring centered at the origin of
the web, for a sufficiently large interval of time.
\end{abstract}

\bigskip

\textit{Keywords}: inverse problems, identification of sources,
spider orb-web model, membrane, infinitesimal vibration.

\section{Introduction}
\label{Intro}

The study of the spider orb-web as biological-mechanical system
has received increasing interest in the last decades. The primary
function of an orb-web is for sensory information aimed at
catching prey, and the analysis of the mechanisms guiding the
spider in prey capture has been – and actually is – one of the
main goals of the scientific investigation in the field. In fact,
even if vision obviously is useful for the detection and
localization of prey, in the majority of spiders prey generated
vibrations are the more important guide \cite{Klarner1982}. In
this paper we consider the inverse problem of localizing a
transversally impacting prey on an orb-web by assuming that the
spider is able to measure the dynamic response in a neighborhood
of the central hub of the web.

In order to better illustrate the motivations of the problem, a
brief overview of the prey detection problem based on dynamic data
is sketched below. The study of prey catching
behavior from the sensory point of view was initiated by Klarner
and Barth \cite{Klarner1982}, see also Masters et al.
\cite{Masters1984} \cite{Masters1984a} for significant contributions on the dynamic
signal transmission through an orb-web. In \cite{Klarner1982}, the authors developed a
comprehensive study on the correlation between prey catching
behavior and vibratory signals generated by prey animals. Based on
one-dimensional models of vibration transmission, typically along
radial threads, it was concluded that the threshold displacement
necessary to elicit predatory behavior decreases with increasing
stimulus frequency, and that airborne vibrations are neither
necessary nor sufficient to elicit the first steps of
prey-catching behavior. The correlation between vibration signals
and spider orientation in the first stage of catching was
investigated by Hergenroder and Barth \cite{Hergenroder1983}. The
experiments conducted by means of a vibrating platform supporting
the spider revealed that, in presence of a bilaterally symmetrical
stimulus combination, the spider turns toward the legs that are
stimulated sooner or with larger amplitude, and the turning angle
corresponds to the response to stimulation of only these legs. The
transmission of natural and artificial vibrations in webs was next
examined using laser Doppler vibrometry by Landolfa and Barth
\cite{Landolfa1996} to determine how the spider discriminates and
localizes stimuli. Using calculation and measurements of the
propagation velocities of waves, the authors provided estimates of
the time-of-arrival differences experienced by the spider placed
at the hub of the web. Again, as in Frolich and Buskirk
\cite{Frohlich1982}, one-dimensional taut string models were used
to calculate propagation velocities in radial threads. However,
Landolfa and Barth \cite{Landolfa1996} were among the first to
recognize that vibrations are not only confined in the stimulated
radius considered in experiments, but instead spread also
laterally across the web, depending on the structural
interconnectivity between the radial threads and the
circumferential threads belonging to the auxiliary or to the
permanent (catching) spiral. One of the merits of their work was
to highlight the limits of one-dimensional models of
interpretation, commonly used at that time, and to suggest the
need to move toward more realistic mechanical models, able to take
into account the two dimensional nature of the spider’s orb web
and its fibred structure.

Along this line of research, significant advances have been
obtained, even in recent years, by using finite element analysis
to deal with realistic web geometry, detailed material
description, and large dynamic deformations. Among other
contributions, we mention the analysis of the role of aerodynamic
drag in the dissipation of prey's energy and in reducing
deterioration of the orb web by Zaera, Soler and Teus
\cite{Zaera2014}; the key effect of the secondary frame in
avoiding excessive stiffness in radial threads by Soler and Zaera
\cite{Soler2016}; the analysis of silk material properties and
propagation of vibrations within webs by Mortimer, Soler, Siviour, Zaera and Vollrath  \cite{Mortimer2016}; the ability of the spider to
adjust web pre-stress and dragline silk stiffness to tune both
transverse and longitudinal wave propagation in the web
\cite{Mortimer2016}. Recently, Otto, Elias and Hatton \cite{Otto2018} studied experimentally the dynamic response of artificial discrete spider webs, using frequency-based substructuring to determine frequency response functions of the system. All the numerical models developed in these investigations offer a remarkable
versatility and accuracy in reproducing the response of the web
under wind loads, prey impacts or other vibratory sources.
However, theoretical models often allow for a deeper insight in
the physical phenomena through the analysis of the underlying
mathematical structure of the governing equations and the
identification of the most relevant parameters that rule the
dynamic response of the web. As we will see in the present paper,
this is particularly relevant in approaching the main inverse
problem afflicting the spider, namely: how to localize a prey by
measuring prey’s induced vibrations in the orb-web?

A first two-dimensional dynamical model of spider orb-web was
proposed by Morassi, Soler and Zaera \cite{Morassi2017}. It is a continuous
mechanical model able to describe small deformations of a spider
orb-web. The actual discrete web, formed by a finite number of
radial and circumferential threads, is approximated by a
continuous elastic membrane on the assumption that the spacing
between threads is small enough. The continuous membrane turns out
to have a specific fibrous structure, which is inherited from the
original discrete web, and it is subject to tensile pre-stress in
the referential configuration. Although the model can be adapted
to reproduce general geometries, for simplicity here we restrict
the attention to circular-shaped webs supported at the boundary,
in which the circumferential threads belong to concentric circles.

Basing on the above mechanical model, in the present paper we
consider the inverse problem of localizing a prey hitting the
orb-web by dynamic response measurements taken near the center of
the web, where the spider is supposed to stay. Prey’s impact is
modelled by a surface force density, whose direction is orthogonal
to the web and intensity is equal to the product between a
function dependent on time, $g=g(t)$, and a function $f=f(x)$
depending only on the spatial position $x$. The function $g$ is
assumed to be known, whereas the function $f$ is the unknown of
the problem. Under the hypothesis of small transverse deformation
- approximation that can be considered acceptable in the interval
of time immediately following the impact - and under the
additional assumption $g(0)\neq0$, we prove the following
uniqueness result: the knowledge of the time-history of the
transverse displacement of the orb-web measured in a circular
annulus centered at the center of the web, having arbitrarily
small thickness and arbitrarily small inner radius, for a
sufficiently large registration time, is sufficient to univocally
determine the function $f$. We refer to Theorem
\ref{teo:Uniqueness} for a precise statement. It is interesting to
note that the minimal data set necessary to ensure uniqueness in
prey’s localization seems to reproduce well the real information
that the spider can collect immediately after the prey’s impact.
In fact, by continuously testing the web, the spider acquires the
dynamical response of the web approximately on a circle centered
at the origin of the web and with radius significantly small with
respect to the web dimensions.

The method used to prove uniqueness is mainly based on the series
representation of the dynamic response of the structure on its
eigenfunction basis, and on the analysis of the almost periodic
distribution that arises from it. This idea was firstly used by
Kawano in \cite{Kawano2013b} to prove uniqueness in the
determination of transverse loads in a vibrating rectangular
elastic plate, and it was recently adapted by Kawano and Morassi
\cite{KM2018} to deal with the analogous inverse problem in
multi-span beams and plates.

{}From the mathematical point of view, the inverse problem of
determining sources in vibrating systems has been the subject of
extensive investigation in the last two decades, and it is not
easy to draw a complete bibliographic overview on the topic. Here,
we limit ourselves to mention some contributions close to the
problem under examination, from which the interested reader can
certainly obtain more information. Yamamoto \cite{Yamamoto1995a}
considered the inverse problem of determining the spatial
component of a source term (e.g., analogous to our function
$f(x)$) in a hyperbolic equation with Laplacian spatial operator
from the observation of the normal derivative of the solution on a
portion of the supported boundary. For a sufficiently large
observation time and basing on exact boundary controllability
methodology, a stability estimate, a reconstruction formula and a
Tikhonov regularization were provided. These results were also
extended to more general hyperbolic equations, whose spatial
second-order operator is given in divergence form with variable
coefficients that, however, should satisfy some positivity
conditions stricter than uniform ellipticity. Best possible
conditional stability estimates for the inverse problem of
determining point waves sources in hyperbolic equations with
Laplacian spatial operator by lateral overdetermined observations
were determined by Komornik and Yamamoto \cite{Komornik2005a}.

As for internal measurements, it is worth to mention that Cheng,
Ding and Yamamoto \cite{Cheng2002} have proved a unique
continuation result for the wave equation in bounded domains
$\Omega\subset \mathbb{R}^n$, $n\leq 3$ under homogeneous
Dirichlet boundary conditions. Their result states that when the
system starts {}from the rest and the loading has the form
$g(t)f(x)$, with $g$ and $f$ regular enough, given any line
segments $\ell\subset L\subset \Omega$, there exists $T_0>0$,
$T_0$ only depending on $\ell$ and $L$, such that if the solution
$u$ of the wave equation vanishes on $\ell$ in the interval of
time $[0,T_0]$, then $f\equiv 0$ in $\Omega$.

{}From the above review it emerges that much of the theory
developed in this area has been concerned with the source
identification problem for homogeneous elastic membranes with
hydrostatic pre-stress state. The spider orb-web model considered
here does not fall within this class of membranes, since its
geometry naturally provides for a peculiar fibrous structure
having an intensification of the density of radial threads towards
the center of the web. In particular, a singularity in the
continuous model arises exactly at the center of the web, e.g.,
the stiffness and the inertia coefficient diverge at that point.
In addition, the presence of the spider in the hub forced to
include a not negligible point mass, of order of $50$ times the
entire mass of the web, at the center of the structure. As it will
be shown in Section \ref{sec:ForcedResponse}, these two peculiar
features of the orb-web characterize the dynamic behavior of the
whole structure.

The plan of the paper is as follows. In Section \ref{sec:Model} we
introduce the continuous membrane model of spider orb-web. The
formulation of the inverse problem and the main uniqueness result,
Theorem  \ref{teo:Uniqueness}, are presented in Section
\ref{sec:MainResult}. Section \ref{sec:ForcedResponse} is devoted
to the analysis of the forced dynamic response of the orb-web via
eigenfunction series representation. A proof of Theorem
\ref{teo:Uniqueness} is presented in Section \ref{sec:ProofMain}.

\section{An overview of a continuous membrane model of a spider orb-web}
\label{sec:Model}

In this section we briefly recall the continuous membrane model
introduced in \cite{Morassi2017} for describing the small
transverse deformation of a spider orb-web. We refer to the above
mentioned paper for more details.

The spider orb-web is considered as a network formed by two
intersecting families of threads which, in a referential
configuration $\mathcal{B}_{\mathcal{K}}$, coincide with radial
directions of a disk centered at the origin $O$ of a Cartesian
reference system (radial threads), and with coaxial circles
centered at $O$ (circumferential threads). No slippage between the
threads belonging to the two families is allowed, and the threads
are supposed to be close enough so that the cable net forms a
continuous structured membrane, in which the two families of
threads may be thought of as coordinate lines on the surface.

The referential placement of a point in the referential state
$\mathcal{B}_{\mathcal{K}}$ is denoted by
\begin{equation}
    \label{eq-1-1}
    \textbf{X}=\textbf{X}(\vartheta_1,
    \vartheta_2)=\vartheta_1(\cos \vartheta_2 \textbf{E}_1 + \sin
    \vartheta_2 \textbf{E}_2 ),
\end{equation}
where the polar coordinates are $\vartheta_1=\rho \in ]0,R]$,
$\vartheta_2 =\vartheta \in[0,2\pi]$, and $\{ \textbf{E}_1,
\textbf{E}_2, \textbf{E}_3=\textbf{E}_1 \times \textbf{E}_2 \}$ is
the canonical basis of $\R^3$, e.g., $\textbf{E}_i \cdot
\textbf{E}_j = \delta_{ij}$, $i,j=1,2,3$, where $\cdot$ is the
scalar product and $\times$ is the vectorial product in $\R^3$.

The actual position $\textbf{x}$ of a particle is given by
\begin{equation}
    \label{eq-2-1}
    \textbf{x}= \textbf{X}+\textbf{u}, \quad \textbf{u}= \sum_{\alpha=1}^2  u^\alpha\textbf{A}_\alpha + u^3
    \textbf{A}_3,
\end{equation}
where $\textbf{u} : \mathcal{B}_{\mathcal{K}} \rightarrow \R^3$ is
the displacement vector field and $
\textbf{A}_\alpha=\frac{\partial \textbf{X}}{\partial
\vartheta_\alpha}= \textbf{X},_\alpha$, $\alpha =1,2$,
$\textbf{A}_3= \textbf{A}_1 \times \textbf{A}_2/|\textbf{A}_1
\times \textbf{A}_2| $.

Our analysis will be developed under the assumption of
infinitesimal deformations, that is, $\textbf{u}$ and the gradient
of $\textbf{u}$, $\nabla \textbf{u}= \frac{\partial \textbf{u}
}{\partial \textbf{X}}$, are small, and high order quantities are
neglected.

Concerning the contact actions inside the orb-web, we assume that
the internal force on a material element taken along a coordinate
line is a tensile force acting in the tangent plane to the
membrane surface, and having direction coincident with the tangent
to the other coordinate curve passing through that element. The
controvariant components of the membrane stress tensor defined in
the actual configuration $\mathcal{B}$ are denoted by $N^{\beta
\alpha}$, with $N^{12}=N^{21}=0$. The equations of equilibrium can
be obtained by imposing the balance of linear momentum for any
portion of the actual configuration $\mathcal{B}$, using Cauchy's
Lemma and the Divergence Theorem. Under the assumption of smooth
tensor and vector fields, we have
\begin{equation}
    \label{eq-3-1}
    \sum_{\alpha=1}^2 N^{\gamma \alpha}|_\alpha + p^\gamma =0,
    \quad  \gamma=1,2, \quad \hbox{in  }
    \mathcal{B},
\end{equation}
\begin{equation}
    \label{eq-3-2}
    \sum_{\alpha, \beta=1}^2 N^{\beta \alpha}b_{\beta \alpha} + p^3 =0, \quad \hbox{in  }
    \mathcal{B},
\end{equation}
where
\begin{equation}
    \label{eq-3-3}
    N^{\gamma \alpha}|_\alpha = N^{\gamma \alpha},_\alpha +
    \sum_{\delta=1}^2  N^{\gamma \delta} \Gamma^{\alpha}_{\delta \alpha} + \sum_{\delta=1}^2 N^{\delta
    \alpha} \Gamma^{\gamma}_{\delta \alpha},
\end{equation}
\begin{equation}
    \label{eq-3-4}
    \Gamma^{\gamma}_{\alpha \beta} = \textbf{a}_{\alpha,\beta} \cdot
    \textbf{a}^\gamma,
    \quad
    b_{\beta \alpha} = \sum_{\gamma=1}^2 b_\alpha^\gamma a_{\gamma \beta},
\end{equation}
\begin{equation}
    \label{eq-3-5}
    a_{\gamma \beta} = \textbf{a}_{\gamma} \cdot \textbf{a}_\beta,
    \quad b_\alpha^\gamma = - \textbf{a}_{3, \alpha} \cdot
    \textbf{a}^\gamma.
\end{equation}
The vectors $\textbf{a}_\alpha$, $\alpha=1,2$, are defined as
$\textbf{a}_\alpha =  \frac{\partial \textbf{x} }{\partial
\vartheta^\alpha}$ and $\textbf{p}= \sum_{\alpha=1}^2 p^\alpha
\textbf{a}_\alpha + p^3
    \textbf{a}_3$ is the continuous surface force density field
    acting on the deformed membrane, inertia
forces included. Moreover, $\textbf{a}_3= \textbf{a}_1 \times
\textbf{a}_2 / |\textbf{a}_1 \times \textbf{a}_2|$ and
$\textbf{a}_{\alpha,\beta}= \frac{\partial \textbf{a}_\alpha
}{\partial \vartheta^\beta}$, $\alpha, \beta=1,2$,
$\textbf{a}_{3,\alpha}= \frac{\partial \textbf{a}_3 }{\partial
\vartheta^\alpha}$, $\alpha=1,2$. Hereinafter, the contravariant
basis $\{ \textbf{a}^1, \textbf{a}^2, \textbf{a}^3 \}$ in a point
of $\mathcal{B}$ is defined as $\textbf{a}^\alpha \cdot
\textbf{a}_\beta = \delta^\alpha_\beta$,
$\textbf{a}^3=\textbf{a}_3$, where $\delta^\alpha_\beta$ is the
Kronecker symbol.

The threads have vanishing shear/bending rigidity, and the
magnitude of the tensile force acting on the threads of the
$\alpha$-family is assumed to depend on the initial tensile
pre-stress and on the elongation in the direction of the $\alpha$
coordinate curve only. In case of linearly elastic material, the
no-vanishing membrane stress components in $\mathcal{B}$ are given
by
\begin{equation}
    \label{eq-4-1}
    N^{11} = d_1 ( {\overline{T}}_1 + \mathcal{A}_1 E_1 \epsilon_1) \frac
     { |a^{11}|^{   \frac{1}{2}}    }
     { |a_{11}|^{   \frac{1}{2}}      },
\end{equation}
\begin{equation}
    \label{eq-4-2}
      N^{22} = d_2 ( {\overline{T}}_2 + \mathcal{A}_2 E_2 \epsilon_2) \frac
     { |a^{22}|^{   \frac{1}{2}}    }
     { |a_{22}|^{   \frac{1}{2}}      },
\end{equation}
where $\mathcal{A}_\alpha$ is the area of the cross-section of a
single thread belonging to the $\alpha$th family and $E_\alpha
>0$ is the Young's modulus of the material. The quantity $\epsilon_\alpha$
is the elongation measure of the threads belonging to the
$\alpha$th family:
\begin{equation}
    \label{eq-4bis-1}
     \epsilon_\alpha = \frac{u_\alpha|_\alpha}{A_{\alpha \alpha}},
     \quad \alpha =1,2,
\end{equation}
where $A_{\alpha \alpha}=\textbf{A}_\alpha \cdot
\textbf{A}_\alpha$ and $u_\alpha|_\alpha = u_{\alpha,\alpha} -
\sum_{\delta=1}^2 \overline{\Gamma}_{\alpha \alpha}^\delta
u_\delta$, $\alpha=1,2$, $\overline{\Gamma}_{\alpha
\alpha}^\delta= \textbf{A}_{\alpha,\beta}\cdot \textbf{A}^\delta$
being the Christoffel symbols defined on
$\mathcal{B}_{\mathcal{K}}$. Moreover, ${\overline{T}}_\alpha>0$
is the tensile pre-stress force acting on a single thread
belonging to the $\alpha$th coordinate curve in the referential
configuration $\mathcal{B}_{\mathcal{K}}$. The quantity $d_\alpha$
expresses the number of threads per unit length crossing the
coordinate curve $\vartheta_\alpha=const$ in $\mathcal{B}$ and
having direction coinciding with the vector $\textbf{a}_\alpha$,
$\alpha=1,2$. Hereinafter, the radial and the circumferential
threads in $\mathcal{B}_{\mathcal{K}}$ are assumed to be equally
spaced both in the plane angle $2 \pi$ and along the radial
direction, respectively, e.g., $\overline{d}_1 =
\frac{C_\rho}{\rho}$ and $\overline{d}_2 =C_\vartheta$, where the
two positive constants $C_\rho$, $C_\vartheta$ are the number of
radial threads per unit plane angle and the number of
circumferential threads per unit length along the radial
direction, respectively. Within the approximation of infinitesimal
deformation, it it possible to prove that
\begin{equation}
    \label{eq-5-1}
     d_1 = \overline{d}_1 \left ( 1 - \left ( u^2,_2 +
     \frac{u^1}{\rho} \right ) \right ),
\end{equation}
\begin{equation}
    \label{eq-5-2}
     d_2 = \overline{d}_2 \left ( 1 - u^1,_1 \right ).
\end{equation}
By inserting \eqref{eq-5-1}, \eqref{eq-5-2} in \eqref{eq-4-1},
\eqref{eq-4-2}, and neglecting higher order terms, we obtain the
linearized constitutive equations of the membrane stresses
\begin{equation}
    \label{eq-5-3}
     N^{11} =  \overline{d}_1 \overline{T}_1 -  \overline{d}_1 \overline{T}_1
    \left ( 2 u^1,_1 + u^2,_2 + \frac{u^1}{\rho} \right ) + \overline{d}_1
    \mathcal{A}_1 E_1 u_1,_1 \ ,
\end{equation}
\begin{equation}
    \label{eq-5-4}
     N^{22} =  \frac{\overline{d}_2 \overline{T}_2}{\rho^2} - \frac{\overline{d}_2
     \overline{T}_2}{\rho^2} \left ( u^1,_1 + 2u^2,_2 +
     2\frac{u^1}{\rho}\right ) +
     \frac{\overline{d}_2 \mathcal{A}_2 E_2 (u_2,_2 +\rho u_1)  }{\rho^4}.
\end{equation}
We conclude this section with a comment on a key aspect of the
model, namely the definition of the pre-stress state of the orb
web. We first notice that the expressions of the membrane stresses
\eqref{eq-5-3}, \eqref{eq-5-4} reduce to pre-stress acting on the
referential configuration $\mathcal{B}_{\mathcal{K}}$ for
vanishing displacement field. More precisely, since the entries of
the second fundamental form of the web surface,
$\overline{b}_{\beta \alpha}$'s, evaluated on
$\mathcal{B}_{\mathcal{K}}$ vanish, and the loads are absent, the
equilibrium equation \eqref{eq-3-2} is identically satisfied, and
equations \eqref{eq-3-1} imply
\begin{equation}
    \label{eq-5bis-1}
     \overline{T}_\vartheta=\overline{T}_\vartheta(\rho), \quad
     \overline{T}_{\rho,\rho}= \xi \overline{T}_\vartheta(\rho),
\end{equation}
with $\xi = \frac{C_\vartheta}{C_\rho}$, where we have defined
$\overline{T}_\rho=\overline{T}_1$,
$\overline{T}_\vartheta=\overline{T}_2$. Hereinafter, we assume
\begin{equation}
    \label{eq-5bis-2}
     \overline{T}_\rho(\rho,\vartheta)|_{\rho=R}=\sigma=constant>0
\end{equation}
and, by symmetry,
\begin{equation}
    \label{eq-5bis-3}
     \overline{T}_\rho = \overline{T}_\rho(\rho).
\end{equation}
Now, we recall that the analysis of the initial pre-stress is
strictly connected with the process followed by spiders in
creating their webs, see \cite{Morassi2017} for details. In brief,
in the first stage of orb web construction, the web configuration
is called \textit{unfinished web}, and the experimental
observations by Wirth and Barth \cite{Wirth1992} support the
hypothesis of proportionality between the circumferential,
$\overline{T}_\vartheta$, and radial, $\overline{T}_\rho$,
pre-stress, namely $\overline{T}_\vartheta(\rho) = k
\overline{T}_\rho (\rho)$, with $ k>0$ constant. Assuming
axially-symmetric referential configuration and writing the
equilibrium equations on $\mathcal{B}_{\mathcal{K}}$, we obtain
\begin{equation}
    \label{eq-6-1}
    \overline{T}_\rho (\rho)=  \widehat{T} \exp (k \xi \rho ), \quad \rho \in [0,R],
\end{equation}
where the radial tensile pre-stress at the center of the web,
$\widehat{T}$, is strictly positive.

In this paper we shall be concerned with the subsequent stage of
the orb-web, the so-called \textit{finished web}, which is
obtained by the spider by removing the auxiliary spiral and
replacing it by the \textit{catching} - or \textit{sticky} -
spiral, which represent the final configuration of the orb-web for
prey catching. Uniform tensile pre-stress in the circumferential
threads of the finished web can be assumed in this stage (see
\cite{Morassi2017}), namely
\begin{equation}
    \label{eq-7-1}
    \overline{T}_\vartheta (\rho) = \mathcal{T} = constant>0, \quad \rho \in
    [0,R],
\end{equation}
which implies
\begin{equation}
    \label{eq-7-2}
    \overline{T}_\rho (\rho) =
    \widehat{T} + \xi \mathcal{T} \rho,
\end{equation}
where $\widehat{T}>0$.

We are now in position to write the equations governing the
infinitesimal transverse forced vibrations of the finished
orb-web, supported at the boundary $\rho=R$ and with a point mass
$M$ attached at $\rho=0$ to model the presence of the spider.
Here, it is implicitly assumed that the effect of the spider on
the stiffness properties of the orb-web is negligible.

By replacing the expressions \eqref{eq-4-1}, \eqref{eq-4-2} in
equation \eqref{eq-3-2}, after linearization, we obtain the
differential equation governing the motion
$U=u^3(\rho,\vartheta,t)$ of the membrane:
\begin{equation}
    \label{eq-7-3}
    \frac{C_\rho}{\rho} \overline{T}_\rho U,_{\rho\rho} +
    C_\vartheta \frac{\overline{T}_\vartheta  }{\rho^2}
    (U,_{\vartheta \vartheta} + \rho U,_\rho) - \widetilde{\gamma}(\rho) U,_ {tt}=- F(\rho,\vartheta,t),
\end{equation}
for every $(\vartheta,\rho,t)\in
[0,2\pi]\times]0,R[\times]0,T_0[$, with $T_0>0$. In the above
equation, the surface mass density $\widetilde{\gamma}$ is given
by
\begin{equation}
    \label{eq-7-4}
   \widetilde{\gamma} = \frac{C_\rho}{\rho}m_\rho + C_\vartheta m_\vartheta,
\end{equation}
where $m_\rho=constant>0$ and $m_\vartheta=constant>0$ is the mass
density per unit length of the radial and circumferential threads,
respectively. The function $F=F(\rho,\vartheta,t)$ is the
transverse force density per unit area, and the pre-stress tensile
state is assumed to satisfy \eqref{eq-7-1}, \eqref{eq-7-2}.

The boundary condition on the external support is
\begin{equation}
    \label{eq-8-1}
   U(\rho, \vartheta,t)|_{\rho=R}=0, \quad (\vartheta,t) \in
   [0,2\pi]\times[0,T_0],
\end{equation}
whereas at $\rho=0$ we provisionally require that $U$, the first
order partial derivatives of $U$, and $U,_{tt}$ are bounded. To
complete the mathematical description of the model, we shall
impose the equilibrium of the out-of-plane components of forces
acting on the small disc $B_\epsilon$, $B_\epsilon
=\{(\rho,\vartheta)| \ \rho < \epsilon, \ \vartheta \in [0,2\pi)\}
\subset \mathcal{B}_{\mathcal{K}}$, centered at the origin $O$ and
with radius $\epsilon$, for every $\epsilon$ small enough. Under
our assumptions, we have
\begin{equation}
    \label{eq-8-3}
    \int_{\partial B_\epsilon}
    \frac{C_\rho}{\rho} \overline{T}_\rho U,_{\rho}
    dS
    =
    \int_{B_\epsilon}
    \widetilde{\gamma}U,_{tt} dA +
    M U,_{tt}(\rho=0,t),
\end{equation}
that is, taking the limit as $\epsilon \rightarrow 0$,
\begin{equation}
    \label{eq-8-4}
    C_\rho \overline{T}_\rho(0) \lim_{\epsilon \rightarrow 0}
    \int_0^{2\pi} U,_{\rho}(\rho=\epsilon,\vartheta,t)d\vartheta =
    MU,_{tt}(\rho=0,t),
\end{equation}
for every $t \in (0,T_0)$.

\section{Inverse problem and main result}
\label{sec:MainResult}

Let us consider the infinitesimal transverse vibration
$U=U(\rho,\vartheta,t)$ of the finished spider orb-web under the
external transverse force $F=F(\rho, \vartheta,t)$. The orb-web
domain is a disk $B_R$ of radius $R$ centered at the origin $O$ of
a Cartesian-coordinate system of the plane. The boundary $\rho=R$
is supported and a concentrated mass $M>0$ is attached at $O$.
Under the hypotheses and notation of Section \ref{sec:Model}, and
assuming
\begin{equation}
    \label{eq-9-1}
    F(\rho, \vartheta, t) = g(t) f(\rho, \vartheta), \quad
    (\rho,\vartheta,t)\in [0,R]\times[0,2\pi]\times[0,T_0],
\end{equation}
the motion $U=U(\rho,\vartheta,t)$ of the orb-web is governed by
the following boundary value problem with initial data:
\begin{center}
\( {\displaystyle \left\{
\begin{array}{lr}
        \frac{C_\rho}{\rho} \overline{T}_\rho U,_{\rho\rho} +
        C_\vartheta \frac{\overline{T}_\vartheta  }{\rho^2}
        (U,_{\vartheta \vartheta} + \rho U,_\rho) - \widetilde{\gamma}(\rho) U,_ {tt}=- g(t)f(\rho,\vartheta),\\
        \quad\quad\quad\quad\quad\quad\quad\quad\quad\quad\quad\quad\quad\quad\quad\quad
    (\rho,\vartheta,t)\in ]0,R[\times[0,2\pi]\times]0,T_0[,
    \vspace{0.25em}\\
        U(\rho, \vartheta,t)|_{\rho=R}=0, \quad t\in [0,T_0],
        \vspace{0.25em}\\
        C_\rho \overline{T}_\rho(0) \lim_{\epsilon \rightarrow 0}
        \int_0^{2\pi} U,_{\rho}(\rho=\epsilon,\vartheta,t)d\vartheta =
        MU,_{tt}(\rho=0,t), \quad t\in ]0,T_0[,
        \vspace{0.25em}\\
        U(\rho, \vartheta,t)|_{t=0}=0, \quad U,_t(\rho,
        \vartheta,t)|_{t=0}=0,  \quad
    (\rho,\vartheta)\in ]0,R[\times[0,2\pi],
        \vspace{0.25em}\\
\end{array}
\right. } \) \vskip -11.0em
\begin{eqnarray}
& & \label{eq-9-2}\\
& & \label{eq-9-4}\\
& & \label{eq-9-5}\\
& & \label{eq-9-6}
\end{eqnarray}
\end{center}
where $U$, the first order partial derivatives of $U$, and
$U,_{tt}$ are bounded at $\rho=0$.

Concerning the direct problem \eqref{eq-9-2}--\eqref{eq-9-6}, if
$g \in C^1([0,T_0])$ and $f \in L^2(B_R)$, then there exists a
unique solution $U\in C^1([0,T_0], L^2(B_R)) $ to \eqref{eq-9-2}--\eqref{eq-9-6}.

Our main result on the determination of the source
$f=f(\rho,\vartheta)$ {}from measurements of the dynamic response
of the finished spider orb-web is the following uniqueness
theorem.

\begin{theorem}
\label{teo:Uniqueness}
For a given $g \in C^1([0,T_0])$, with $g(0) \neq 0$, let
$U_{f_1}$, $U_{f_2}$ be the solution to
\eqref{eq-9-2}--\eqref{eq-9-6} corresponding to the distributed
force field $f_1 \in L^2(B_R)$, $f_2\in L^2(B_R)$, respectively.
Then, the existence of any set of the type
$\omega=]\rho-\epsilon,\rho+\epsilon[ \times [0,2\pi] \times
[0,\tau_0[$, $0<\tau_0 <T_0$, $0<\epsilon<\rho <
\frac{R}{2}$, so that
$U_{f_1}(t,\rho,\theta)=U_{f_2}(t,\rho,\theta)$ in $\omega$,
implies that $f_1\equiv f_2$ in $B_R$.
\end{theorem}

The proof of Theorem \ref{teo:Uniqueness} is given in Section
\ref{sec:ProofMain}.

\section{Forced dynamic response of the orb-web}
\label{sec:ForcedResponse}

Our proof of the uniqueness theorem strongly relies on the series
representation of the solution to \eqref{eq-9-2}--\eqref{eq-9-6}
in terms of the eigensolutions of the orb-web. Therefore, the
present section is mainly devoted to the study of the free
vibration problem of the system.

\subsection{Eigenvalue problem}
\label{EigProblem}

Setting in \eqref{eq-9-2}, with $f\equiv 0$,
\begin{equation}
    \label{eq-11-1}
   U(\rho, \vartheta,t)=w(\rho,\vartheta)y(t),
\end{equation}
we can separate the spatial variables $(\rho, \vartheta)$ {}from
the time variable $t$, obtaining
\begin{equation}
    \label{eq-11-2}
   y'' + \lambda y =0, \quad t>0,
\end{equation}
and
\begin{equation}
    \label{eq-11-3}
    \frac{C_\rho}{\rho} \overline{T}_\rho w,_{\rho\rho} +
    C_\vartheta \frac{\overline{T}_\vartheta  }{\rho^2}
    (w,_{\vartheta \vartheta} + \rho w,_\rho) + \lambda \widetilde{\gamma}(\rho) w=0, \quad
    (\rho, \vartheta)\in ]0,R[\times[0,2\pi],
\end{equation}
where $\lambda \in \R^+$ is the eigenvalue and $w$ the
corresponding eigenfunction. We look for a non-trivial solution to
\eqref{eq-11-3} of the form
\begin{equation}
    \label{eq-12-1}
   w(\rho,\vartheta)=u(\rho)\Theta(\vartheta).
\end{equation}
Since $w$ is a periodic function of $\vartheta$, we infer the
following eigenvalue problem for $\Theta(\vartheta)$:
\begin{center}
\( {\displaystyle \left\{
\begin{array}{lr}
        \Theta''  +\nu^2 \Theta=0,
        \quad \vartheta \in ]0,2\pi[,
    \vspace{0.25em}\\
    \Theta(0)=\Theta (2\pi),
        \vspace{0.25em}\\
    \Theta'(0)=\Theta'(2\pi),
        \vspace{0.25em}\\
\end{array}
\right. } \) \vskip -8.5em
\begin{eqnarray}
& & \label{eq-12-2}\\
& & \label{eq-12-3} \\
& & \label{eq-12-4}
\end{eqnarray}
\end{center}
which admits the infinite sequence of eigenpairs
\begin{equation}
    \label{eq-12-5}
   \nu_n^2=n^2, \quad \Theta^{(n)}(\vartheta)=A_n \cos(n\vartheta)
   +B_n \sin(n\vartheta), \quad n=0,1,2,...
\end{equation}
Recalling \eqref{eq-5bis-1}, if $n=0$, then
$\Theta^{(0)}(\vartheta)=constant \neq 0$, and we have the
classical Sturm-Liouville equation
\begin{equation}
    \label{eq-12-6}
    (C_\rho\overline{T}_\rho {u^{(0)}}')' + \lambda^{(0)}\gamma u^{(0)}    =0,
        \quad \rho \in (0,R),
\end{equation}
where
\begin{equation}
    \label{eq-12-7}
    \gamma \equiv \rho \widetilde{\gamma}=C_\rho m_\rho + \rho
    C_\vartheta m_\vartheta.
\end{equation}
In this case, the corresponding eigenfunctions $w^{(0)}$ are
axially-symmetric functions.

If $n \geq 1$, then a Coulomb-like singularity appears at
$\rho=0$:
\begin{equation}
    \label{eq-13-1}
    (C_\rho\overline{T}_\rho {u^{(n)}}')' +  \lambda^{(n)} \gamma {u^{(n)}}  =
    n^2 \frac{C_\vartheta \overline{T}_\vartheta }{\rho} {u^{(n)}},
        \quad \rho \in ]0,R[.
\end{equation}
The eigenvalue problem is completed by assigning the boundary
conditions to the equations \eqref{eq-12-6} ($n =0$) and
\eqref{eq-13-1} ($n \geq 1$). At $\rho=R$, the boundary is
supported and then
\begin{equation}
    \label{eq-13-2}
    {u^{(n)}}(R)=0, \quad n \geq 0.
\end{equation}
To derive the boundary condition at $\rho=0$, let us use the
dynamic equilibrium equation \eqref{eq-8-4}. Replacing
$U=w^{(n)}(\rho,\vartheta)y(t)$ in \eqref{eq-8-4}, with
$w^{(n)}=u^{(n)}\Theta^{(n)}$, $n \geq 0$, using \eqref{eq-11-2}
and assuming $u^{(n)}(0)$, ${u^{(n)}}'(0)$ bounded, we have
\begin{equation}
    \label{eq-13-3}
    C_\rho \overline{T}_\rho(0) {u^{(n)}}'(0)
    \int_0^{2\pi} \Theta^{(n)}(\vartheta)d\vartheta = - M w^{(n)}(\rho=0)\lambda^{(n)},
\end{equation}
Now, we can distinguish two cases.

Case i) $n\geq 1$. Then $\int_0^{2\pi}
\Theta^{(n)}(\vartheta)d\vartheta=0$, and the end condition at
$\rho=0$ is
\begin{equation}
    \label{eq-13-4}
    u^{(n)}(0)=0.
\end{equation}
Case ii) $n=0$. Then $\int_0^{2\pi}
\Theta^{(n)}(\vartheta)d\vartheta=2\pi \cdot const$, with $const
\neq 0$, and the end condition at $\rho =0$ becomes
\begin{equation}
    \label{eq-13-5}
    2\pi C_\rho \overline{T}_\rho(0) {u^{(0)}}'(0)= -\lambda^{(0)}M u^{(0)}(0).
\end{equation}

\subsection{Eigenpair class $n=0$}
\label{n=0}

The eigenvalue problem is as follows:
\begin{center}
\( {\displaystyle \left\{
\begin{array}{lr}
        (C_\rho \overline{T}_\rho {u^{(0)}}')' + \lambda^{(0)}\gamma u^{(0)} =0,
        \quad \rho \in ]0,R[,
    \vspace{0.25em}\\
    2\pi C_\rho \overline{T}_\rho(0) {u^{(0)}}'(0)= -\lambda^{(0)}M u^{(0)}(0),
        \vspace{0.25em}\\
    u^{(0)}(R)=0.
        \vspace{0.25em}\\
\end{array}
\right. } \) \vskip -8.0em
\begin{eqnarray}
& & \label{eq-14-1}\\
& & \label{eq-14-2} \\
& & \label{eq-14-3}
\end{eqnarray}
\end{center}
The eigenvalue problem \eqref{eq-14-1}--\eqref{eq-14-3} has an
infinite sequence of real, simple eigenvalues
$\{\lambda^{(0)}_m\}_{m=1}^\infty$ such that $0 < \lambda^{(0)}_1
<\lambda^{(0)}_2<...$ and $\lim_{m \rightarrow +\infty}
\lambda^{(0)}_m =+\infty$. Moreover, the following asymptotic
eigenvalue estimate holds:
\begin{equation}
    \label{eq-14-4}
    \sqrt{ \lambda^{(0)}_m   } =C^{(0)}m\pi + O(m^{-1}), \quad
    \hbox{as } m \rightarrow +\infty,
\end{equation}
where $C^{(0)}>0$ is a suitable constant independent of $m$.

The proof of the above properties follows {}from a classical
result due to Fulton \cite{Fulton1977}. Let us sketch the main
steps of the proof. In order to simplify the notation, let us
rewrite \eqref{eq-14-1}--\eqref{eq-14-3} as
\begin{center}
\( {\displaystyle \left\{
\begin{array}{lr}
        \frac{d}{dy}\left ( p \frac{dv}{dy} \right ) +\lambda r v   =0,
        \quad \rho \in ]0,R[,
    \vspace{0.25em}\\
    \left ( p \frac{dv}{dy} \right )(0) = - \lambda
    \mathcal{M}v(0),
        \vspace{0.25em}\\
    v(R)=0,
        \vspace{0.25em}\\
\end{array}
\right. } \) \vskip -8.0em
\begin{eqnarray}
& & \label{eq-14A1-1}\\
& & \label{eq-14A1-2} \\
& & \label{eq-14A1-3}
\end{eqnarray}
\end{center}
where $y=\rho$, $v(y)=u^{(0)}(\rho)$, $p(y)=C_\rho
\overline{T}_\rho(\rho)$, $r(y)=\gamma(\rho)$,
$\lambda=\lambda^{(0)}$, $\mathcal{M}= \frac{M}{2\pi}$. By the
Liouville transformation
\begin{equation}
    \label{eq-14A1-4}
    x = \frac{1}{J} \int_0^y \sqrt{  \frac{r(s)}{p(s)}  }ds, \quad
    J= \int_0^R \sqrt{  \frac{r(s)}{p(s)}  }ds,
\end{equation}
\begin{equation}
    \label{eq-14A1-5}
    A(x)=(r(y)p(y))^{\frac{1}{2}}, \quad u(x)=v(y),
\end{equation}
the problem \eqref{eq-14-1}--\eqref{eq-14-3} takes the
impedance-like form
\begin{center}
\( {\displaystyle \left\{
\begin{array}{lr}
        \frac{d}{dx}\left ( A \frac{du}{dx} \right ) +\mu A u   =0,
        \quad x \in ]0,1[,
    \vspace{0.25em}\\
    \left ( A \frac{du}{dx} \right )(0) = - \mu \frac{\mathcal{M}}{J}
    u(0),
        \vspace{0.25em}\\
    u(1)=0,
        \vspace{0.25em}\\
\end{array}
\right. } \) \vskip -8.0em
\begin{eqnarray}
& & \label{eq-14A1-6}\\
& & \label{eq-14A1-7} \\
& & \label{eq-14A1-8}
\end{eqnarray}
\end{center}
where $\mu= \lambda J^2$. Letting $A(x)=a^2(x)$, $z(x)=a(x)u(x)$,
the problem \eqref{eq-14A1-6}--\eqref{eq-14A1-8} takes the
Sturm-Liouville canonical form
\begin{center}
\( {\displaystyle \left\{
\begin{array}{lr}
        \frac{d^2 z}{dx^2} + \mu z =qz,
        \quad x \in ]0,1[,
    \vspace{0.25em}\\
    \left ( a \frac{dz}{dx}- z \frac{da}{dx}  \right )(0) = - \mu \frac{\mathcal{M}}{J}
    \frac{z(0)}{a(0)},
        \vspace{0.25em}\\
    z(1)=0,
        \vspace{0.25em}\\
\end{array}
\right. } \) \vskip -8.0em
\begin{eqnarray}
& & \label{eq-14A2-1}\\
& & \label{eq-14A2-2} \\
& & \label{eq-14A2-3}
\end{eqnarray}
\end{center}
with $q(x)= \frac{1}{a} \frac{d^2a}{dx^2}$ in $[0,1]$, $q \in
C^0([0,1])$. Now, we can use the analysis by Fulton
\cite{Fulton1977} (Section 4, case 4) to conclude that there
exists a countable infinite family or real, simple eigenvalues
$\mu_m$ to \eqref{eq-14A2-1}--\eqref{eq-14A2-3} with
\begin{equation}
    \label{eq-14A2-4}
    \sqrt{ \mu_m  } =m\pi + O(m^{-1}), \quad
    \hbox{as } m \rightarrow +\infty,
\end{equation}
which implies the asymptotic estimate \eqref{eq-14-4}. Moreover,
\begin{equation}
    \label{eq-14A2-5}
    z_m(x) = \sqrt{2}  \sin(m\pi(1-x)) + O(m^{-1}), \quad
    \hbox{as } m \rightarrow +\infty.
\end{equation}

\subsection{Eigenpair classes $n\geq 1$}
\label{n>0}

Let $n$, $n \geq 1$, be given. The eigenvalue problem along the
radial coordinate $\rho$ is as follows:
\begin{center}
\( {\displaystyle \left\{
\begin{array}{lr}
        (C_\rho \overline{T}_\rho {u^{(n)}}')' + \lambda^{(n)}\gamma u^{(n)} =
        n^2 \frac{C_\vartheta \overline{T}_\vartheta }{\rho} {u^{(n)}},
        \quad \rho \in ]0,R[ ,
    \vspace{0.25em}\\
    u^{(n)}(0)=0,
        \vspace{0.25em}\\
    u^{(n)}(R)=0.
        \vspace{0.25em}\\
\end{array}
\right. } \) \vskip -8.0em
\begin{eqnarray}
& & \label{eq-15-1}\\
& & \label{eq-15-2} \\
& & \label{eq-15-3}
\end{eqnarray}
\end{center}
Also in the present case, by applying two Sturm-Liouville
transformations in sequence, equation \eqref{eq-15-1} is reduced
to impedance form, and then to the standard canonical form
\begin{equation}
    \label{eq-15-4}
   y''(x) + \mu^{(n)} y(x) = (\widetilde{q}(x) + V(x))y(x), \quad x
   \in ]0,1[,
\end{equation}
where $\mu^{(n)}= G^2 \lambda^{(n)}$, $G=  \int_0^R  \left (
\frac{\gamma(s)}{ \overline{T}_\rho(s) } \right )^{  \frac{1}{2}}
ds$, $\widetilde{q}$ is a regular function in $[0,1]$, and the
potential $V(x)$ contains the singularity, namely $|V(x)| \leq
\frac{C}{x}$ in $]0,1]$ for some positive constant $C$. Now, the
analysis developed by Carlson \cite{Carlson1994} shows that, for
every $n \geq 1$, the Dirichlet eigenvalue problem associated to
\eqref{eq-15-4} has an infinite countable sequence of eigenpairs
$\{\mu^{(n)}_m, y_m^{(n)}(x)\}_{m=1}^\infty$, and the following
asymptotic estimates hold
\begin{equation}
    \label{eq-15-5}
    \sqrt{ \mu^{(n)}_m  } = m\pi + O(m^{-1+\epsilon}), \quad
    \hbox{as } m \rightarrow +\infty,
\end{equation}
\begin{equation}
    \label{eq-15-5-bis}
    y_m^{(n)}(x) = \sqrt{2} \sin (m\pi x) + O(m^{-1+\epsilon}), \quad
    \hbox{as } m \rightarrow +\infty,
\end{equation}
where $0<\epsilon<1$. {}From \eqref{eq-15-5} one can deduce that
\begin{equation}
    \label{eq-15-5-ter}
    \sqrt{ \lambda^{(n)}_m  } =C^{(n)} m\pi + O(m^{-1+\epsilon}), \quad
    \hbox{as } m \rightarrow +\infty,
\end{equation}
where $C^{(n)}$ is a suitable constant independent of $m$.

\subsection{Series solution of the forced problem}
\label{SeriesSolution}

Let us formally derive a series representation of the solution of
the forced dynamic problem \eqref{eq-9-2}--\eqref{eq-9-6}. We look
for a solution $U=U(\rho,\vartheta,t)$ expressed in series of
eigenfunctions of the orb-web as
\begin{multline}
    \label{eq-16-1}
    U(\rho,\vartheta,t) = \sum_{m=1}^\infty c_m^{(0)}(t)
    u_m^{(0)}(\rho)+ \\
    + \sum_{n=1}^\infty \sum_{m=1}^\infty
    \left (
    c_m^{(n)}(t) \cos (n\vartheta) + d_m^{(n)}(t) \sin (n\vartheta)
    \right ) u_m^{(n)}(\rho),
\end{multline}
where $u_m^{(0)}=u_m^{(0)}(\rho)$ and, for every given $n \geq 1$,
$u_m^{(n)}=u_m^{(n)}(\rho)$ are the eigensolutions to
\eqref{eq-14-1}--\eqref{eq-14-3} and
\eqref{eq-15-1}--\eqref{eq-15-3}, respectively. The eigensolutions
are normalized as follows:
\begin{equation}
    \label{eq-orthonormalization}
    < u_m^{(0)}, u_i^{(0)}>_{\gamma,M} =
    \frac{\delta_{mi}}{\lambda^{(0)}_i},
    \quad
    < u_m^{(n)}, u_i^{(n)}>_{\gamma} =
    \frac{\delta_{mi}}{\lambda^{(n)}_i},
    \quad i,m,n \geq 1,
\end{equation}
where $\delta_{mi}$ denotes the Kronecker's symbol, and the scalar
products $< \cdot, \cdot >_{\gamma,M}$, $< \cdot, \cdot
>_{\gamma}$ are defined as
\begin{equation}
    \label{eq-scalprod-1}
    < h_1, h_2 >_{\gamma,M}=
    \int_0^R \gamma(\rho) h_1(\rho) h_2(\rho)d\rho + \frac{M}{2\pi}
    h_1(0)h_2(0),
\end{equation}
\begin{equation}
    \label{eq-scalprod-2}
    < h_1, h_2 >_{\gamma}=
    \int_0^R \gamma(\rho) h_1(\rho) h_2(\rho)d\rho,
\end{equation}
for every smooth function $h_i: [0,R] \rightarrow \R$, $i=1,2$.

It is convenient to rewrite equation \eqref{eq-9-2} multiplying
both sides by $\rho$, that is
\begin{equation}
    \label{eq-16-2}
    (C_\rho \overline{T}_\rho U,_\rho)_\rho - \gamma U,_ {tt}=
    - \frac{C_\vartheta \overline{T}_\vartheta }{\rho } U,_{\vartheta \vartheta} - \rho g(t)f(\rho, \vartheta).
\end{equation}
Let us derive the initial value problem for the function
$c_m^{(0)}(t)$, $m \geq 1$. We start by considering the
equilibrium condition \eqref{eq-9-5}. By using the series
expansion \eqref{eq-16-1} of $U$, and noting that $\int_0^{2\pi}
\cos(n\vartheta)d\vartheta = \int_0^{2\pi}
\sin(n\vartheta)d\vartheta =0$ for every $n \geq 1$,
$u_m^{(n)}(0)=0$ for every $n,m \geq 1$, condition \eqref{eq-9-5}
becomes
\begin{equation}
    \label{eq-c0-A}
    \sum_{m=1}^\infty
    2\pi C_\rho \overline{T}_\rho (0) {u_m^{(0)}}'(0) c_m^{(0)}(t)
    =
    M \sum_{m=1}^\infty
    {c_m^{(0)}}''(t) u_m^{(0)}(0).
\end{equation}
Next, we replace the series expression of $U$ in \eqref{eq-9-2},
we multiply by $u_i^{(0)}$ and we integrate on $[0,2\pi] \times
[0,R]$. Since $\int_0^{2\pi} \cos(n\vartheta)d\vartheta =
\int_0^{2\pi} \sin(n\vartheta)d\vartheta =0$ for every $n \geq 1$,
the terms involving $u_m^{(n)}$ vanish, and we have
\begin{multline}
    \label{eq-c0-B}
    2\pi
    \sum_{m=1}^\infty
    \left (
    c_m^{(0)}(t) \int_0^R (C_\rho \overline{T}_\rho {u_m^{(0)}}')'
    u_i^{(0)} d\rho
    -
    {c_m^{(0)}}''(t) \int_0^R \gamma u_m^{(0)} u_i^{(0)} d\rho
    \right )
    =
    \\
    = - \int_0^{2\pi} \int_0^R g(t) f(\rho,\vartheta)u_i^{(0)}
    \rho d\rho d\vartheta.
\end{multline}
Integrating by parts on the first integral and using
\eqref{eq-c0-A}, the left hand side of \eqref{eq-c0-B} becomes
\begin{multline}
    \label{eq-c0-C}
    2\pi
    \sum_{m=1}^\infty
    \left (
    c_m^{(0)}(t)
    \left (
    - C_\rho \overline{T}_\rho (0) {u_m^{(0)}}'(0) u_i^{(0)}(0)
    -
    \int_0^R C_\rho \overline{T}_\rho {u_m^{(0)}}' {u_i^{(0)}}' d\rho
    \right )
    \right.
    -
    \\
    -
    \left. {c_m^{(0)}}''(t) \int_0^R \gamma u_m^{(0)} u_i^{(0)} d\rho
    \right ) =
    \\
    =
    - \sum_{m=1}^\infty
    \left (
    {c_m^{(0)}}''(t)
    \left (
    2\pi \int_0^R \gamma {u_m^{(0)}} {u_i^{(0)}} d\rho + M {u_m^{(0)}}(0)
    {u_i^{(0)}}(0)
    \right ) -
    \right.
    \\
    - \left.
     c_m^{(0)}(t) 2\pi \int_0^R C_\rho \overline{T}_\rho {u_m^{(0)}}' {u_i^{(0)}}' d\rho
     \right ).
\end{multline}
Using the orthogonality condition \eqref{eq-orthonormalization}
(left) and by the weak formulation of the eigenvalue problem, we
finally have
\begin{equation}
    \label{eq-c0-D}
   {c_m^{(0)}}''(t)+ \lambda^{(0)}_m c_m^{(0)}(t) =
        \lambda^{(0)}_m \mathcal{F}^{(0)}_m g(t), \quad t>0,
\end{equation}
where
\begin{equation}
    \label{eq-c0-E}
    \mathcal{F}^{(0)}_m = \frac{1}{2\pi} \int_0^{2\pi} \int_0^R f
    u_m^{(0)} \rho d\rho d\vartheta.
\end{equation}
The differential equation \eqref{eq-c0-D} is completed with the
initial values
\begin{equation}
    \label{eq-c0-F}
    {c_m^{(0)}}(0)=0, \quad {c_m^{(0)}}'(0)=0, \quad \hbox{for
    every } m \geq 1.
\end{equation}
By proceeding similarly, the initial value problems for
${c_m^{(n)}}(t)$, ${d_m^{(n)}}(t)$ are given by
\begin{center}
\( {\displaystyle \left\{
\begin{array}{lr}
        {c_m^{(n)}}''(t)+ \lambda^{(n)}_m c_m^{(n)}(t) = \lambda^{(n)}_m
         \mathcal{F}^{(n)}_{Cm} g(t), \quad t>0,
    \vspace{0.25em}\\
    c_m^{(n)}(0)=0,
        \vspace{0.25em}\\
    {c_m^{(n)}}'(0)=0,
        \vspace{0.25em}\\
\end{array}
\right. } \) \vskip -8.0em
\begin{eqnarray}
& & \label{eq-17-2}\\
& & \label{eq-17-3} \\
& & \label{eq-17-4}
\end{eqnarray}
\end{center}
\begin{center}
\( {\displaystyle \left\{
\begin{array}{lr}
        {d_m^{(n)}}''(t)+ \lambda^{(n)}_m d_m^{(n)}(t) = \lambda^{(n)}_m
         \mathcal{F}^{(n)}_{Sm} g(t), \quad t>0,
    \vspace{0.25em}\\
    d_m^{(n)}(0)=0,
        \vspace{0.25em}\\
    {d_m^{(n)}}'(0)=0,
        \vspace{0.25em}\\
\end{array}
\right. } \) \vskip -8.0em
\begin{eqnarray}
& & \label{eq-17bis-2}\\
& & \label{eq-17bis-3} \\
& & \label{eq-17bis-4}
\end{eqnarray}
\end{center}
where, for $m,n \geq 1$,
\begin{equation}
    \label{eq-c0-I}
    \mathcal{F}^{(n)}_{Cm} = \frac{1}{\pi} \int_0^{2\pi} \int_0^R f
    u_m^{(n)} \cos(n\vartheta) \rho d\rho d\vartheta,
\end{equation}
\begin{equation}
    \label{eq-c0-L}
    \mathcal{F}^{(n)}_{Sm} = \frac{1}{\pi} \int_0^{2\pi} \int_0^R f
    u_m^{(n)} \sin(n\vartheta) \rho d\rho d\vartheta.
\end{equation}
The solution to the above initial value problems is given by the
Duhamel integral representation
\begin{equation}
    \label{eq-17-5}
    c_m^{(0)}(t)= \sqrt{\lambda^{(0)}_m } \mathcal{F}^{(0)}_m
    \int_0^t \sin \left ( \sqrt{\lambda^{(0)}_m } (t-\tau) \right ) g(\tau)
    d\tau, \quad t>0, \ m \geq 1,
\end{equation}
\begin{equation}
    \label{eq-18-1}
    c_m^{(n)}(t)= \sqrt{\lambda^{(n)}_m } \mathcal{F}^{(n)}_{Cm}
    \int_0^t \sin \left ( \sqrt{\lambda^{(n)}_m } (t-\tau) \right )g(\tau)
    d\tau, \quad t>0, \ m,n \geq 1,
\end{equation}
\begin{equation}
    \label{eq-18-2}
    d_m^{(n)}(t)= \sqrt{\lambda^{(n)}_m } \mathcal{F}^{(n)}_{Sm}
    \int_0^t \sin \left ( \sqrt{\lambda^{(n)}_m } (t-\tau) \right ) g(\tau)
    d\tau, \quad t>0, \ m,n \geq 1.
\end{equation}
Therefore, the series representation of the solution $U$ becomes
\begin{multline}
    \label{eq-18-3}
    U(\rho,\vartheta,t) = \sum_{m=1}^\infty \sqrt{\lambda^{(0)}_m } \mathcal{F}^{(0)}_m
    \int_0^t \sin \left ( \sqrt{\lambda^{(0)}_m } (t-\tau) \right )g(\tau)
    d\tau
    \cdot u_m^{(0)}(\rho)+ \\
    + \sum_{n=1}^\infty \sum_{m=1}^\infty
    \sqrt{\lambda^{(n)}_m }  \left (
    \mathcal{F}^{(n)}_{Cm} \cos (n\vartheta) + \mathcal{F}^{(n)}_{Sm} \sin (n\vartheta)
    \right ) \times\\
    \times \int_0^t \sin \left ( \sqrt{\lambda^{(n)}_m } (t-\tau) \right )
    g(\tau) d\tau \cdot u_m^{(n)}(\rho).
\end{multline}
We notice that to justify the series expansion \eqref{eq-16-1} we
need to prove the completeness of the eigenfunctions of the
orb-web problem $\{w_m^{(n)}(\rho,\vartheta)\}_{n\geq 0, m \geq
1}$, where
\begin{center}
\( {\displaystyle \left\{
\begin{array}{lr}
        w_m^{(0)}(\rho,\vartheta)=u_m^{(0)}(\rho), \quad m \geq 1,
    \vspace{0.25em}\\
        w_m^{(n)}(\rho,\vartheta)=(A_m^{(n)} \cos(n\vartheta) + B_m^{(n)}
        \sin(n\vartheta))u_m^{(n)}(\rho),\quad m,n\geq1,
    \vspace{0.25em}\\
\end{array}
\right. } \) \vskip -6.0em
\begin{eqnarray}
& & \label{eq-18-4}\\
& & \label{eq-18-5}
\end{eqnarray}
\end{center}
with $A_m^{(n)}$, $B_m^{(n)}$ real constants. More precisely, we
shall prove the following result. Let $h \in L^2(B_R)$. If
$\int_{B_R} hw_m^{(n)}=0$ for every $n\geq 0$ and every $m\geq 1$,
then $h=0$ in $B_R$. Using density arguments, we can assume that
$h$ is regular enough such that its Fourier series
\begin{equation}
    \label{eq-19-1}
    h(\rho,\vartheta)= a_0(\rho) + \sum_{k=1}^\infty
    (a_k(\rho)\cos(k\vartheta) +b_k(\rho)\sin(k\vartheta))
\end{equation}
is uniformly convergent in $B_R$. Choosing $n=0$ and using
\eqref{eq-19-1}, we obtain
\begin{equation}
    \label{eq-19-2}
    \int_0^R a_0(\rho)\rho u_m^{(0)}(\rho) d\rho=0, \quad
    \hbox{for every } m\geq 1.
\end{equation}
By the completeness of the eigenfunctions
$\{u_m^{(0)}\}_{m=1}^\infty$ we have $a_0 \equiv 0$ in $[0,R]$.
Next, for a given $n$, $n \geq 1$, and choosing $B_m^{(n)}=0$,
$A_m^{(n)}=1$ for every $m \geq 1$, {}from \eqref{eq-19-1} we have
\begin{equation}
    \label{eq-19-4}
    \int_0^R a_n(\rho)\rho u_m^{(n)}(\rho) d\rho=0, \quad
    \hbox{for every } m\geq 1,
\end{equation}
which implies $a_n \equiv 0$ in $[0,R]$ by the completeness of the
eigenfunctions $\{u_m^{(n)}\}_{m=1}^\infty$. One can proceed
similarly to prove that $b_n \equiv 0$ in $[0,R]$, and therefore
$h=0$ in $B_R$.

Expression \eqref{eq-18-3} shows that the spatial force
distribution $f$ affects the dynamic response of the orb-web by
means of the coefficients $\{ \mathcal{F}^{(0)}_m\}_{m=1}^\infty$,
$\{\mathcal{F}^{(n)}_{Cm},
\mathcal{F}^{(n)}_{Sm}\}_{m,n=1}^\infty$. We conclude this section
by showing that the knowledge of all these coefficients allows to
determine $f$. More precisely, setting
\begin{equation}
    \label{eq-F1}
    \mathcal{F}(\rho,\vartheta)=\frac{\rho
    f(\rho,\vartheta)}{\gamma(\rho)},
\end{equation}
we show that the following eigenfunction expansion holds true:
\begin{multline}
    \label{eq-F2}
    \mathcal{F}(\rho,\vartheta) = \sum_{m=1}^\infty \lambda^{(0)}_m
    \mathcal{F}^{(0)}_m u_m^{(0)}(\rho)+ \\
    + \sum_{n=1}^\infty \sum_{m=1}^\infty
    \lambda^{(n)}_m \left (
    \mathcal{F}^{(n)}_{Cm} \cos (n\vartheta) + \mathcal{F}^{(n)}_{Sm} \sin (n\vartheta)
    \right ) u_m^{(n)}(\rho).
\end{multline}

It is enough to prove that \eqref{eq-F1}--\eqref{eq-F2} imply
\eqref{eq-c0-E}, \eqref{eq-c0-I}, \eqref{eq-c0-L}. We first
consider $\mathcal{F}^{(0)}_m$, for a given integer $m$, $m \geq
1$. Let us take the scalar product $<\cdot,\cdot>_{\gamma,M}$ of
both sides of \eqref{eq-F2} with $u_k^{(0)}$, for given $k \geq
1$, and then integrate on $[0,2\pi]$ with respect to the variable
$\vartheta$. Noticing that $\int_0^{2\pi}
\cos(n\vartheta)d\vartheta=\int_0^{2\pi}
\sin(n\vartheta)d\vartheta=0$ for every $n \geq 1$, and
$u_m^{(n)}(0)=0$ for every $m,n \geq 1$, we have
\begin{multline}
    \label{eq-F3}
    \int_0^{2\pi}
    \left (
    \int_0^R \mathcal{F} \gamma u_k^{(0)} d\rho +
    \frac{M}{2\pi}
    \mathcal{F} (\rho=0) u_k^{(0)}(0)
    \right )
    d\vartheta
    =
    \\
    =
    \int_0^{2\pi}
    \left (
    \sum_{m=1}^\infty
    \lambda^{(0)}_m
    \mathcal{F}^{(0)}_m
    \left (
    \int_0^R \gamma u_m^{(0)} u_k^{(0)} d\rho
    + \frac{M}{2\pi} u_m^{(0)}(0)  u_k^{(0)}(0)
    \right )
    \right )
    d\vartheta.
\end{multline}
Using the orthogonality condition \eqref{eq-orthonormalization}
(left) and observing that the boundary term at $\rho=0$ on the
left hand side of \eqref{eq-F3} vanishes (e.g., $\mathcal{F}
(\rho=0)=0$), we obtain the desired expression
\begin{equation}
    \label{eq-F3bis}
    \frac{1}{2\pi} \int_0^{2\pi} \int_0^R f u_m^{(0)}\rho d\rho
    d\vartheta =\mathcal{F}^{(0)}_m.
\end{equation}
Next, we take the scalar product  $<\cdot,\cdot>_\gamma$ of both
sides of \eqref{eq-F2} with the function
$u_i^{(k)}\cos(k\vartheta)$, for given $k$ and $i$, $k,i \geq 1$,
and we integrate on $[0,2\pi]$ with respect to the variable
$\vartheta$. We have
\begin{multline}
    \label{eq-F4}
    \int_0^{2\pi}
    \int_0^R \mathcal{F} \gamma u_i^{(k)} d\rho
    d\vartheta
     =
    \sum_{m=1}^\infty
    \lambda^{(0)}_m
    \mathcal{F}^{(0)}_m
    \int_0^R \gamma u_m^{(0)} u_i^{(k)} d\rho
    \int_0^{2\pi} \cos(k\vartheta)
    d\vartheta
    +
    \\
    +
     \sum_{n=1}^\infty \sum_{m=1}^\infty
    \lambda^{(n)}_m \left (
    \mathcal{F}^{(n)}_{Cm} \int_0^{2\pi}\cos (n\vartheta)\cos (k\vartheta)d\vartheta
     + \mathcal{F}^{(n)}_{Sm} \int_0^{2\pi}\sin (n\vartheta)\cos (k\vartheta)d\vartheta
    \right ) \cdot\\
    \cdot \int_0^R \gamma u_m^{(n)} u_i^{(k)} d\rho.
\end{multline}
Noticing that $\int_0^{2\pi} \cos(k\vartheta)d\vartheta=0$, by the
orthogonality of the trigonometric functions and by the
orthogonality condition \eqref{eq-orthonormalization} (right), we
have
\begin{equation}
    \label{eq-F5}
    \frac{1}{\pi} \int_0^{2\pi} \int_0^R f u_i^{(k)} \cos(k\vartheta) \rho d\rho
    d\vartheta =\mathcal{F}^{(k)}_{Ci}.
\end{equation}
The expression of $\mathcal{F}^{(k)}_{Si}$ can be determined
similarly.

\section{Proof of the main theorem}
\label{sec:ProofMain}

Before giving the proof of Theorem \ref{teo:Uniqueness}, we state
some auxiliary results and definitions.
\begin{definition}
Given a bounded set $S\subset \mathbb{R}$, with positive measure,
the one dimensional Paley-Wiener space $PW_S$ is defined as
    \begin{displaymath}
        PW_S=\{\widehat{F}| \ F\in \mathrm{L}^{2},\,   \mathrm{supp}(F)\subset
        S\},
    \end{displaymath}
where $\widehat{F}$ is the Fourier transform of the function $F$.
\end{definition}

Note that by the Paley-Wiener Theorem, the space $PW_S$ can be
characterized by the functions that can be extended to an entire
function in the whole $\mathbb{C}$, whose Fourier transform is
compactly supported in $S$.

\begin{definition}
A sequence $(\mu_n)_{n\in\mathbb{N}}\subset \mathbb{R}$ is
\emph{uniformly discrete} if there is $\epsilon>0$ such that for
any pair $(m,n)\in\mathbb{N}\times \mathbb{N}$ with $m \neq n$,
$|\mu_m-\mu_n|>\epsilon$.
\end{definition}

\begin{definition}
    The \emph{upper uniform density} of a uniformly discrete set $\mathbb{M}$ is defined by
    \begin{displaymath}
        D(\mathbb{M})=\lim_{c\rightarrow+\infty}\max_{a\in\mathbb{R}}\frac{\#(\mathbb{M}\cap]a,a+c[)}{c}.
    \end{displaymath}
\end{definition}

\begin{definition}
An indexed set $\mathbb{M} \doteq
(\mu_n)_{n\in\mathbb{N}}\subset\mathbb{R}$ is an interpolation set
for $PW_S$, $S\subset \mathbb{R}$ bounded with positive measure,
if for every sequence $(c_n)_{n\in\mathbb{N}}\subset \ell^{2}$
there is $\phi\in PW_S$ such that $\phi(\mu_n)=c_n$ for every $
\mu_n\in \mathbb{M}$.
\end{definition}

In the case of an interval $]a_1,a_2[\subset\mathbb{R}$, there is
a sufficient condition for $\mathbb{M}$ to be an interpolating
sequence \cite{Olevskii2009}. Precisely:
\begin{equation}
\label{eq:Kahane}
    D(\mathbb{M}) < \frac{1}{2\pi} (a_2-a_1)\Rightarrow \mbox{$\mathbb{M}$ is an interpolating set of $PW_{]a_1,a_2[}$}.
\end{equation}
The following lemma, which is a special case of a result that  can
be found in \cite{Kawano2011}, is directly related to the goal we
are seeking.
\begin{lemma}\label{teo:LemaExtensao}

Consider $F(t,x)=\sum_{n\in\mathbb{N}} A_n(x) \mathrm{e}^{- i \mu_n
t}$, $F: [0,T]  \times  \Omega  \rightarrow \R$,  with $(A_n(x))_{n\in\mathbb{N}}\subset \ell^{2}$. Suppose that
the sequence $(\mu_n)_{n\in\mathbb{N}}$ is uniformly discrete, with
$\mu_n=\mathcal{O}(n)$ as $n\rightarrow +\infty$.

There exists  $T_0>0$ such that if there is an interval $]0,T[
\subset \mathbb{R}$ with $T\geq T_0$ and there is a set
$\Omega_0$, $\Omega_0\subset \Omega$, such that  $F(t,x)=0$ in
$]0,T[\times \Omega_0$, then $F\equiv 0$ in $\Omega$.
\end{lemma}

We recall that with the internal product
\begin{displaymath}
\langle h_1, h_2 \rangle_{\mathcal{M}}= \int_{0}^{2\pi} \langle
h_1(\cdot,\vartheta), h_2(\cdot,\vartheta) \rangle_{\gamma, M}\,
d\vartheta,
\end{displaymath}
the family of functions $\{u_m^{(0)}, u_m^{(n)}\cos(n \vartheta),
u_m^{(n)}\sin(n \vartheta) \}_{m,n\in\N}$ is an orthogonal family.
Moreover, we have already proved that the space generated with
this family is complete. We call this Hilbert space $\mathcal{M}$.

Our unknown function $f$ belongs to $L^{2}(B_R)$ in the sense that
using the natural norm of $L^{2}(B_R)$,
\begin{displaymath}
\int_{0}^{2\pi} \int_{0}^{R} |f(\rho,\vartheta)|^2 \rho \,d\rho\,d\vartheta< +\infty.
\end{displaymath}

Now, since $f \in L^{2}(B_R)$, the series \eqref{eq-F2} of
$\mathcal{F}(\rho,\vartheta)$ converges in the space
$\mathcal{M}$. In fact:
\begin{multline*}
\sum_{m\in\N} \lambda_m^{(0)} {\mathcal{F}_m^{(0)}}^2 + \sum_{m,n\in\N} \lambda_m^{(n)} \mathcal{F}_{Cm}^2 + \sum_{m,n\in\N} \lambda_m^{(n)} \mathcal{F}_{Sm}^2 = \langle \mathcal{F}, \mathcal{F} \rangle_{\mathcal{M}} = \\
= \left\langle \frac{\rho f}{\gamma}, \frac{\rho f}{\gamma} \right\rangle_{\mathcal{M}}= \frac{1}{2\pi} \int_{0}^{2\pi} \left\langle \frac{\rho f}{\gamma}, \frac{\rho f}{\gamma} \right\rangle_{\gamma,M}\, d\vartheta = \frac{1}{2\pi} \int_{0}^{2\pi}\int_0^R \frac{\rho}{\gamma} |f|^2 \rho \, d\rho\,d\vartheta\\
\leq K \int_{0}^{2\pi} \int_0^R  |f|^2 \rho \, d\rho\,d\vartheta < +\infty,
\end{multline*}
because $f \in L^{2}(B_R)$.

The above calculation also shows that the three sequences
$(\sqrt{\lambda_m^{(n)}} \mathcal{F}_{Cm}^{(n)} )_{m\in\N}$,
$(\sqrt{\lambda_m^{(n)}} \mathcal{F}_{Sm}^{(n)} )_{m\in\N}$ and  $(\sqrt{\lambda_m^{(0)}} \mathcal{F}_m^{(0)} )_{m\in\N}$  belong
to $\ell^2$.

\begin{lemma} \label{teo:Sequencesinl2}
    For every $\varphi \in C_c^\infty (]0,R[)$, the sequences
    \begin{displaymath}
    (\mathcal{F}_m^{(0)}\sqrt{\lambda_m^{(0)}} u_m^{(0)}\varphi)_{m\in\N}, \,
    (\mathcal{F}_{C_m}^{(n)} \sqrt{\lambda_m^{(n)}} u_m^{(n)}\varphi)_{m\in\N}, \,
    (\mathcal{F}_{S_m}^{(n)} \sqrt{\lambda_m^{(n)}}  u_m^{(n)}\varphi)_{m\in\N}
    \label{eq:coeffUniqueness}
    \end{displaymath}
    belong to the sequence space $\ell^2$.
\end{lemma}
\begin{proof}
The test function $\varphi\in C_c^{\infty}(]0,R[)$ can be expressed
as the series $\varphi=\sum_{m\in\N} C_m u_{m}^{(n)}$, for all
$n\in\N\cup \{0\}$, where $C_{m}=\lambda_{m}^{(n)}
u_{m}^{(n)}\varphi$. Since $\varphi\in L^{2}(]0,R[)$,  we conclude
that $(u_{m}^{(n)}\varphi\sqrt{\lambda_m})_{m\in\N}\in \ell^2$,
$\forall n\in\N\cup \{0\}$.

Since we already know that the sequences $(\sqrt{\lambda_m^{(0)}}
\mathcal{F}_m^{(0)} )_{m\in\N}$, $(\sqrt{\lambda_m^{(n)}}
\mathcal{F}_{Cm}^{(n)} )_{m\in\N}$ and $(\sqrt{\lambda_m^{(n)}}
\mathcal{F}_{Sm}^{(n)} )_{m\in\N}$ all belong to $\ell^2$, the
conclusion of the lemma follows readily.
\end{proof}

\begin{lemma}
\label{teo:NonNullEigenvector}
Let $u^{(n)}_m$, $n\in \N\cup\{0\}$, $m\in \N$, be an
eigenfunction either of the problem
\eqref{eq-14-1}--\eqref{eq-14-3} or of the problem
\eqref{eq-15-1}--\eqref{eq-15-3}. For every $m,n$, there is no
open set $\Omega\subset [0,R]$ such that $u^{(n)}_m|_\Omega=0$.
\end{lemma}

We now provide a proof of the uniqueness result on identification
of distributed forces over the spider orb-web stated in Theorem
\ref{teo:Uniqueness}.

\begin{proof}[Proof of Theorem \ref{teo:Uniqueness}]
We rewrite the series expression \eqref{eq-18-3} of $U$, $U \in C^1 ([0,T_0], \mathcal{M})$, to emphasize its convolution form:
\begin{equation}
\label{eq:Convolution}
    U(\rho, \vartheta, t) = \int_0^t g(t-\tau)
    \mathcal{U}(\rho,\vartheta,\tau) \, d\tau,
\end{equation}
where, to simplify the notation, we have defined
\begin{multline*}
\mathcal{U}(\rho,\vartheta,\tau) = \left\{
\sum_{m=1}^{\infty} \mathcal{F}_m^{(0)}\sqrt{\lambda_m^{(0)}}\sin(\sqrt{\lambda_m^{(0)}} \tau) u_m^{(0)}(\rho)+\right.\\
\left. +\sum_{n=1}^{\infty}\sum_{m=1}^{\infty}
\sqrt{\lambda_m^{(n)}} \left( \mathcal{F}_{Cm}^{(n)}
\cos(n\vartheta) + \mathcal{F}_{Sm}^{(n)} \sin(n\vartheta) \right)
\sin(\sqrt{\lambda_m^{(n)}} \tau) u_m^{(n)}(\rho) \right\}.
\end{multline*}
Since the problem is linear, it suffices to prove that if $f \in
L^2 (]0,R[) \times ]0,2\pi[)$ is a distributed force in
\eqref{eq-9-2}--\eqref{eq-9-6}, then
\begin{displaymath}
    U_{f}=0 \quad \text{in $\omega$} \Rightarrow f\equiv 0.
\end{displaymath}
By hypothesis, there exists a set $ \omega=]\rho_1,\rho_2[\times
[0,2\pi[ \times [0,\tau_0[$ such that $U|_{\omega}=0$. Therefore,
\begin{displaymath}
    \label{Volterra}
    0=\frac{\partial U}{\partial t}= g(0)\mathcal{U}(\rho,\vartheta,t) +
    \int_0^t g^\prime(t-\tau) \mathcal{U}(\rho,\vartheta,\tau)\, d\tau,
\end{displaymath}
in $\omega$. Since $g(0)\neq 0$, \eqref{Volterra} is a Volterra
integral equation of the second kind, and then
\begin{equation}
\label{eq:FuncaoNula}
    \mathcal{U}(\rho,\vartheta,t)=0
\end{equation}
for $(\rho,\vartheta,t)\in \omega$.

Since $\{ 1, \vartheta\mapsto \sin(n\vartheta), \vartheta
\mapsto\cos(n\vartheta), \ n \in \mathbb{N} \}$ is an orthogonal
family of functions in $[0,2\pi[$, we conclude that for all
$(\rho,t)\in ]\rho_1,\rho_2[\times[0,\tau_0[$, we have
\begin{eqnarray}
\sum_{m=1}^{\infty}\mathcal{F}_m^{(0)}\sqrt{\lambda_m^{(0)}}\sin(\sqrt{\lambda_m^{(0)}} t) u_m^{(0)}(\rho)&=&0,
\label{eq:FirstUniqueness}\\
\sum_{m=1}^{\infty}\sqrt{\lambda_m^{(n)}} \mathcal{F}_{Cm}^{(n)}
\sin(\sqrt{\lambda_m^{(n)}} t) u_m^{(n)}(\rho)&=&0,\quad
\hbox{for every } n\in \mathbb{N},\label{eq:SecondUniqueness}\\
\sum_{m=1}^{\infty}\sqrt{\lambda_m^{(n)}}
 \mathcal{F}_{Sm}^{(n)}
\sin(\sqrt{\lambda_m^{(n)}} t) u_m^{(n)}(\rho)&=&0, \quad
\hbox{for every } n\in \mathbb{N}, \label{eq:ThirdUniqueness}
\end{eqnarray}
after multiplying \eqref{eq:FuncaoNula} in turn by the constant function $1$,
and by the functions $\sin(n\vartheta)$ and $\cos(n\vartheta)$, and integrating over
$]0,2\pi[$.

All three equations above can be extended oddly in a natural way
to the interval $]-\tau_0,\tau_0[$ and can be expressed in trigonometric
form $\sum_{n\in\mathbb{N}} A_n \mathrm{e}^{- i \mu_n t}$.
 We apply each of these equations to test functions
$\varphi\in C^{\infty}_c(]\rho_1,\rho_2[)$ in radial space variable.
The coefficients $A_n$ multiplying the trigonometric exponentials
become respectively
\begin{equation}
\label{eq:coeffUniqueness}
    \mathcal{F}_m^{(0)}\sqrt{\lambda_m^{(0)}} u_m^{(0)}\varphi, \,
     \mathcal{F}_{Cm}^{(n)} \sqrt{\lambda_m^{(n)}} u_m^{(n)}\varphi, \,
     \mathcal{F}_{Sm}^{(n)} \sqrt{\lambda_m^{(n)}}
     u_m^{(n)}\varphi.
\end{equation}

Putting $\mu_j=\sqrt{\lambda_j}$, from the asymptotic estimates
\eqref{eq-14-4} and  \eqref{eq-15-5}, we have that the sequence
$(\mu^{(n)}_j)_{j\in\N}$ is uniformly discrete, for every $n\in
\N\cup\{0\}$. Moreover, for every $n\in \N\cup\{0\}$, we have
$\mu_j=\mathcal{O}(j)$, for $j\rightarrow +\infty$.

By Lemma \ref{teo:LemaExtensao} and  Lemma
\ref{teo:Sequencesinl2}, the coefficients
\eqref{eq:coeffUniqueness}  multiplying the trigonometric
exponentials are all null.

Since no eigenfunction $u_m^{(n)}$ is identically null in
$]\rho_1,\rho_2[$, $\forall m\in \mathbb{N}$, $\forall n\in
\mathbb{N}\cup \{0\}$, the conclusion is that
$\mathcal{F}_m^{(0)}=0$, $\mathcal{F}_{Cm}^{(n)}=0$ and
$\mathcal{F}_{Sm}^{(n)}=0$, for every $n,m\in \mathbb{N}$, and
then, by using \eqref{eq-F1} and \eqref{eq-F2}, $f \equiv 0$.
\end{proof}

\section{Conclusions}
\label{Conclusions}

In this paper we have considered the inverse problem of
determining the position of a prey hitting the orb-web by
measurements of the dynamic response taken near the center of the
web, where the spider is supposed to stay. The orb web was
described by means of a continuous membrane model with a specific
fibrous structure to take into account the presence of radial and
circumferential threads.

For an axially-symmetric orb-web supported at the boundary and
undergoing infinitesimal transverse deformation, we proved a
uniqueness result for the spatial distribution of the loading in
terms of measurements of the dynamic response taken on an
arbitrarily small and thin ring centered at the origin of the web,
for a sufficiently large interval of time.

A forthcoming paper will be devoted to the numerical
implementation of a reconstruction method suggested by the present uniqueness result.

\section*{Acknowledgments}
The authors thank the financial support {}from Fapesp -- Sao Paulo
Research Foundation -- (Proc. $2017/07189-2$, Visiting Researcher
Program, and Proc. $2017/06452-1$). The second author gratefully acknowledge also the
support of the National Research Project PRIN $2015TT JN95$
'Identification and monitoring of complex structural systems'.

\bibliographystyle{plain}
\bibliography{References}

\end{document}